\documentclass{article}
\usepackage{eurosym}
\usepackage{graphics,epsfig,amsfonts,amsmath}
\usepackage{amsmath, amsthm}
\usepackage{amsfonts}
\usepackage{amssymb}
\usepackage{graphicx}
\usepackage{float}
\usepackage{color}
\usepackage{enumerate}
\usepackage{comment}
\usepackage{cite}

\newcommand{\sgn}{\text{sgn}}
\newtheorem{theorem}{Theorem}[section]

\begin{document}

\title{Boundedness vs Unboundedness of A Noise Linked to Tsallis q-Statistics: The Role of The Overdamped Approximation}
\author{Dario Domingo$^{1,2}$, Alberto d'Onofrio$^{3}$, Franco Flandoli$^{2}$}
\date{}
\maketitle

\begin{abstract}
An apparently ideal way to generate continuous bounded stochastic processes
is to consider the stochastically perturbed motion of a point of small mass in an
infinite potential well, under overdamped approximation. 
Here, however, we show that the aforementioned procedure can be fallacious and lead to 
incorrect results.
We indeed provide a counter-example concerning one of the most employed bounded noises, 
hereafter called Tsallis-Stariolo-Borland (TSB) noise,
which admits the well known Tsallis q-statistics as stationary density.
In fact, we show that for negative values of the Tsallis parameter $q$ (corresponding 
to sufficiently large diffusion coefficient of the stochastic force), the motion
resulting from the overdamped approximation is unbounded. We then
investigate the cause of the failure of Kramers first
type approximation, and we formally show that the solutions of the
full Newtonian non-approximated model are bounded, following the physical intuition. Finally, we provide a new family of
bounded noises extending the TSB noise, the boundedness of whose solutions we
formally show.
\end{abstract}

\smallskip $^{1}$\textit{School of Mathematics, University of Leeds,
Woodhouse Lane, Leeds LS2 9JT (United Kingdom)\newline
}

\smallskip $^{2}$\textit{Department of Mathematics, University of Pisa,
Largo Bruno Pontecorvo 5, 56127 Pisa (Italy)}\newline

\smallskip $^{3}$\textit{International Prevention Research Institute, 95
Cours Lafayette, 69006 Lyon (France)\newline
}

\bigskip


\section{Introduction}
In mathematical biophysics, the influence of extrinsic sources of stochasticity in otherwise deterministic biological systems 
is very frequently taken into account in an elementary way. 
Indeed, the deterministic dynamical system (valid in the absence of the above-mentioned sources) 
is often perturbed by allowing one or more of its parameters to stochastically fluctuate via a white noise or a colored Gaussian perturbation.
This approach is very interesting and often allows to make analytical or partially-analytical inferences, but it can lead to artifacts, 
sometimes not perceived by modelers.\\
For example, as stressed in \cite{Alberto}, modeling the stochastic fluctuations affecting an anti-tumor therapy 
by means of a white noise means that the model allows the possibility that the therapy adds tumor cells instead of killing them: 
a very gross artifact. Indeed, by denoting with $Y$ the tumor size,
with $r(Y)$ its net growth rate, and with $\theta>0$ the anti-tumor cytotoxic therapy, one gets the following mathematical model:
\begin{equation}
dY = r(Y)dt - \theta Y(dt+ \sigma dB) .
\end{equation}
This implies that, in the realizations of the stochastic process, quite frequently the term $\theta Y(dt+ \sigma dB) $ is negative:
in biological terms, tumor cells would be added!
Moreover, this modeling approach also allows an excessive instantaneous growth of the therapy term, 
which is another equally important biological artifact. The same problems hold if one models the perturbation by an Ornstein-Uhlenbeck (OU) colored noise.\\
Finally, white noise perturbations can only be applied to parameters on which a system depends linearly, thus severely limiting the scope of stochastic perturbations. 
By adopting OU noises the scenario only slightly improves. \\
An alternative strategy which is becoming increasingly important consists in modeling parametric perturbations by bounded noises (see \cite{Alberto} and references therein). 
Indeed, bounded-noise perturbations allow to maintain the positiveness and boundedness of perturbed parameters, 
and also allow to model the fluctuations of all parameters of a model, 
even those on which a system depends nonlinearly.\\
Of course a key factor in bounded-noise based stochastic modeling is the characterization 
of the stationary probability density function (PDF) of the stochastic bounded perturbation.
As far as this fundamental issue is concerned, from a statistical physics point of view, a natural and important extension 
of the Gaussian PDF is the Tsallis q-statistics, which is at the basis of the non-extensive statistical mechanics \cite{TsallisReview}. 
Indicating with $Z$ a random variable that follows the Tsallis q-statistics, its PDF is as follows \cite{TsallisReview,PratoTsallis,Tsallis1995}:
\begin{equation}
\rho_q(z) = A_q \left( 1 - \frac{1-q}{3-q} \left(\frac{z}{\sigma}\right)^2 \right)^{1/(1-q)},
\end{equation}
where $\sigma>0$, $q$ is a real number smaller than $3$, and $A_q$ is the normalization constant. 
This distribution has  some noteworthy properties \cite{TsallisReview}:
\begin{enumerate}[i)]
 \item $\lim_{q\rightarrow 1} \rho_q(z) = N(0,\sigma^2 (3-q)/2 )$;
 \item  for $q>1$, the distribution for $z \gg \sigma$ is a power law (with diverging second moment if $q \in (5/3,3)$);
 \item for $q<1$, $Z$ is bounded, with indeed $\rho_q(z)$  equal to zero outside $\left(-\sigma \sqrt{\frac{3-q}{1-q}}, \sigma \sqrt{\frac{3-q}{1-q}}\right)$;
 \item for $ q \rightarrow -\infty$, the distribution of $z$ tends to the uniform distribution in $(-1,1)$.
\end{enumerate}
Thus, the Tsallis thermostatistical theory is not only able to unify the Gaussian and power-law Levy behaviors, 
as stressed in \cite{TsallisReview}, but it also describes an important class of bounded stochastic behaviors.\\
Stariolo \cite{Stariolo} (see also \cite{Borland}) investigated the problem of identifying an overdamped 
stochastic dynamics in the phase space that has the Tsallis q-distribution as equilibrium PDF. 
For the sake of precision, Stariolo defined a quite general family of SDEs 
that depend on symmetric unspecified potentials $V(x)$ and leading to a generalization of the Tsallis q-statistics. 
Namely, the Tsallis q-statistics is obtained in the case of quadratic $V(x)$.\\
For such quadratic potentials, the resulting non-Gaussian bounded process (as well as the case $q>1$) 
has been investigated in a series of influential papers \cite{WioToral,FuentesSR,WioSR,WioPath,Don2010,Baura} 
showing that the departure from the Gaussian PDF in the noise induces remarkable effects in 
noise-induced transitions and in stochastic resonance \cite{WioToral,Wio,Don2010,FuentesSR,WioSR}. 
This process is sometimes called Tsallis-Borland process \cite{Alberto}, although it should be more precisely 
called the Tsallis-Stariolo-Borland (TSB) process, as we will do in the following.\\
The Stariolo family of SDEs can be put in an even more general framework: 
the motion of a material point of position $y$ in overdamped regime (thus neglecting the effect of the mass),
under the action of a deterministic force $F(y)$ and a stochastic white noise force $\xi(t)$:
\begin{equation} 
\dot y = F(y) + \eta \xi(t),  
\end{equation}
where $\eta$ is a positive constant, $F(y)$ is such that 
\begin{equation}
 \lim_{y\rightarrow -1^+} F(y)=+\infty \quad \text{and} \quad \lim_{y\rightarrow 1^-} F(y)=-\infty \,,
\end{equation}
and the potential $U(y)$ associated with $F(y)$ is such that
\begin{equation}
\lim_{y\rightarrow \pm 1}U(y)=+\infty.  
\end{equation}
This suggests that an ideal physical `recipe' to generate bounded noises is to consider the overdamped motion of a point
in a potential well of infinite height, under the perturbation of a
stochastic external ``white noise'' force. This is a particular limit case of
the classical problem of statistical physics studied by
Kramers in its hugely influential paper published in 1940 \cite{Kramers}.
\newline
After recalling in Section~\ref{Sec_basic} the basic physical interpretation of the TSB process, 
we formally prove in Section~\ref{Sec_q<0} that the above-mentioned recipe, quite full of appeal, can be fallacious:
we show that the non-Gaussian TSB process undergoes a stochastic bifurcation at $q=0$, 
with the process being in fact unbounded for $q<0$. 
In Section~\ref{Sec_second_order}, however, we prove that the associated paradoxical physical scenario 
is only apparent, since  by taking into account the mass of the point the resulting motion remains bounded. 
In other words, the unboundedness of the TSB stochastic process for $q<0$ 
is a mathematical artifact caused by the overdamped approximation.
Finally, in Section~\ref{Sec_alpha_family}, we propose a family of deterministic forces that generalize the TSB noise and that induce bounded motions
in the overdamped approximation (as well as in the non-approximate case).

\section{Basic notions}\label{Sec_basic}
Let us consider a material point $P$ of mass $m$ and position $(x, y, z)$ on
which 1D forces act along the $x$-axis, and that at time $t=0$ is not moving
in the $(y,z)$ plane. Thus, its subsequent motion will only be along the
$x$-axis with $(y(t),z(t))=(y(0),z(0))$ for all $t\geq 0$. 
Suppose the forces acting on $P$ be as follows:
\begin{enumerate}[i)]
 \item A linear viscous force : 
\begin{equation}
F_v (t) = -\gamma \dot x (t).
\end{equation}
For the sake of notation simplicity we set henceforward: 
$$ \gamma = 1. $$
\item A stochastic white noise force: 
\begin{equation}
F_s(t) =  \sqrt{2\beta}\,\xi(t).
\end{equation}
\item\label{cons} A conservative force: 
\begin{equation}
F(x) = F_l(x) + F_r(x), 
\end{equation}
where $F_l(x)$ and $F_r(x)$ are two repulsive forces centered, respectively,
at $x = -1$ and at $x = 1$, and such that their potentials
(denoted as $U_l(x)$ and $U_r(x)$) are infinite at the respective centres of
repulsion. In other words, we require that:
\begin{equation}
\lim_{ x\rightarrow -1^{-}}F_l(x) = -\infty, \quad \lim_{ x\rightarrow -1^{+}}F_l(x) = +\infty  ;
\end{equation}
\begin{equation}
\lim_{ x\rightarrow 1^{-}}F_r(x) = -\infty, \quad \lim_{ x\rightarrow 1^{+}}F_r(x) = +\infty ;
\end{equation}
\vspace{-1mm}
\begin{equation}
\lim_{ x\rightarrow -1}U_l(x) = +\infty, \quad \lim_{ x\rightarrow 1}U_r(x) = +\infty .
\end{equation}
\end{enumerate}
The Newton's equation for the motion of $P$ thus reads as follows: 
\begin{equation}\label{newton1}
m \ddot x = - \dot x  + F_l(x) + F_r(x) + \sqrt{2\beta}\xi(t) ,
\end{equation}
If the mass of the point P is much smaller than the viscous constant rate $\gamma$, i.e. in our notation: 
\begin{equation}
m \ll 1 ,
\end{equation}
then one can adopt the first type Kramers approximation
by neglecting the contribution of the acceleration. This yields: 
\begin{equation}  \label{Kapp}
\dot x  = F_l(x) + F_r(x) + \sqrt{2 \beta}\xi(t).
\end{equation}
As a consequence, the first order SDE (\ref{Kapp}) is an apparently
excellent simple model to define a bounded noise $x(t)$: a stochastically
perturbed material point $P$ moving in an infinite potential well and  subject to strongly viscous force, thus remaining confined in $(-1,1)$. \newline
For example, assuming: 
\begin{equation}  \label{FlTSBE}
F_l(x) =  \frac{\sgn(x+1)}{|x+1|}
\end{equation}
and
\begin{equation}  \label{FrTSBE}
F_r(x) =  \frac{\sgn(x-1)}{|x-1|}
\end{equation}
one gets the Tsallis-Stariolo-Borland Equation (TSBE) \cite{Stariolo,Borland, Alberto,WioToral,FuentesSR,WioSR,WioPath,Don2010}: 
\begin{equation}  \label{TSBE}
\dot x  = \frac{-2x}{1-x^2} + \sqrt{2 (1-q)}\xi(t),
\end{equation}
where $q$ is a real parameter smaller than 1.

\section{Unboundedness of the TSB Noise for $q<0$}\label{Sec_q<0}

Equation (\ref{TSBE}) is a SDE of the form
\begin{equation}\label{SDE}
 \dot x  = \varphi(x) + \sigma(x) \, \xi(t)\,,
\end{equation}
where the drift and the (constant) diffusion are given by: 
\begin{equation}\label{phi}
\varphi(x)= -\frac{2x}{1- x^2}\,,   \qquad \sigma(x) \equiv \sqrt{2 \,(1-q)}\,, \;\; q \leq 1\,.
\end{equation}
Consider the process $x(t)$ solution of~(\ref{TSBE}), with deterministic initial condition $x_0 \in I=(-1,1)$.
Based on the physical model that generates TSBE, the process $x$ apparently
satisfies $x(t) \in I$ for all times $t\geq 0$.
In order to carry out a formal investigation of this point, it is convenient to introduce the first exit time of the process from $I$,
denoted as $T(x_0)$:
\begin{equation}\label{T(x_0)}
T(x_0):=\inf \left\{ t>0 \;\big|\;\, x(t)\notin I, \; x(0)=x_0 \right\}\,. 
\end{equation}
The aim is to study the conditions, if any, under which the random time $T(x_0)$ is almost surely infinite.
\\
Some regularity conditions on the coefficients of the SDE, such as the positiveness of $\sigma^2$ and the local integrability
of $(1+|\varphi|)/\sigma^2$, assure that the process leave any compact subinterval $[a,b] \subseteq I$ with probability one \cite{KS}. 
Note that, in the case of TSBE, both conditions are fulfilled. We can therefore consider the almost surely \textit{finite} random time
\begin{equation}\label{Tab}
T_{a,b}(x_0):=\inf \left\{ t>0 \;\big|\;\, x(t)\notin [a,b], \; x(0)=x_0 \right\}
\end{equation}
and study the case $a\rightarrow -1$, $b \rightarrow 1$.
Of course, the process $x$ at the time $T:=T_{a,b}(x_0)$ will always satisfy either $x(T)=a$ or $x(T)=b$. 
The probabilities of these two mutually exclusive events can be expressed as follows \cite{KS}:
\begin{equation} \label{P(a)}
P\big( x(T)=a \big) = \frac{s(b) - s(x_0)}{s(b) - s(a)}\,, \qquad 
P\big( x(T)=b \big) = \frac{s(x_0) - s(a)}{s(b) - s(a)}\,,
\end{equation}
where $s(x)$ is the \emph{scale function} associated with the process $x(t)$. 
For a process satisfying an SDE as~(\ref{SDE}), the scale function takes the form
\begin{equation}  \label{s}
s(x)= \int_c^x \exp \left( - \int_c^y 2\:\frac{\varphi(z)}{\sigma^2(z)}\;
dz\;\right)\,dy\,, \quad x\in I\,,
\end{equation}
where $c$ is any point in $I$. Notice that, despite being $s(x)$ dependent on the choice of $c$,
the RHSs of~(\ref{P(a)}) are both independent of it. 
Indeed, the relationship between two scale functions $s_1$ and $s_2$
obtained by choosing different constants $c$ is of the type ${s_2}(x) = \alpha \,s_1(x) + \beta$,
for some constants $\alpha, \beta$.
\\
In the case~(\ref{phi}) of TSBE, the scale function where $c=0$ thus reads as follows:
\begin{equation}\label{s_TSBE}
s(x)=\int_0^x {\left( 1- z^2 \right)}^{-1/(1-q)}\,dz\,. \vspace{1mm} 
\end{equation}
In particular, notice that the value of $q$ affects
whether $| s(\pm 1)|$ is finite or infinite.
The following result about the behavior of TSBE near the boundaries in the case $q \geq 0$ can now be proved.
\begin{theorem}\label{q>0}
 Consider TSBE~(\ref{TSBE}) with $q \in [0,1]$. Then, for any initial condition $x_0 \in I=(-1,1)$, the solution
 $x(t)$ remains in $I$ for all times $t\geq 0$, with probability one.
 In terms of the random time $T(x_0)$ introduced in~(\ref{T(x_0)}), we can write
 $$
 P\big( T(x_0) = \infty \big) =1\,.
 $$
\end{theorem}
\begin{proof}
Consider any compact subinterval $[a,b] \subseteq I$ containing $x_0$.
Up to time $T_{a,b}(x_0)$ (eq.~(\ref{Tab})), a strong solution of~(\ref{TSBE}) exists and is unique, because of the
Lipschitz property of its coefficients in $[a,b]$.
Such solution can also be uniquely extended up to the explosion time $T(x_0)$, because the coefficients remain
locally Lipschitz on the whole interval $I$.
\\
Now observe that
\begin{equation}
P\Big( x\big(T_{a,b}(x_0)\big) =a \Big) \leq P \left( \inf_{t < T(x_0)} x(t) \; \leq a \right),
\end{equation}
i.e.
\begin{equation}\label{P_ab}
\frac{s(b) - s(x_0)}{s(b) - s(a)} \leq P \left( \inf_{t < T(x_0)} x(t) \; \leq a \right)
\end{equation}
given the first identity in~(\ref{P(a)}).
Also, from~(\ref{s_TSBE}),
\begin{equation}
s(1) = \int_0^1 {\left( 1- z^2 \right)}^{-1/(1-q)}\,dz = \infty \quad \text{ since } q\geq0 \,.
\end{equation}
Therefore, by letting $b$ tend to $1$ in~(\ref{P_ab}), we get $P \big( \inf_{t < T(x_0)} x(t) \leq a \big) = 1$.
Similarly, one has
$
P \big( \sup_{t < T(x_0)} x(t) \geq b \big) =~1$, since $s(-1)=-\infty$.
Since this holds for any $a,b \in I$, we have actually shown that
\begin{equation}\label{1}
P \left( \inf_{t < T(x_0)} \!\!\!x(t) \, = -1 \right) = P \left( \sup_{t < T(x_0)} \!\!\!x(t) \, = 1 \right) = 1
\end{equation}
under the hypothesis $q \geq 0$.\\
We can now show that the event $A :=\{ T(x_0) < \infty \}$ has null probability.
Indeed, on $A$, we can either have that $\inf_{t < T(x_0)} x(t) =-1$ or that $\sup_{t < T(x_0)} x(t) =1$.
Therefore, given~(\ref{1}), we have
\begin{equation}
P(A) = P\left( A\; \cap \; \left\{  \inf_{t < T(x_0)} \!\!\!x(t) \, = -1 \right\} \right) +
       P\left( A\; \cap \; \left\{  \sup_{t < T(x_0)} \!\!\!x(t) \, = 1 \right\} \right) = P(A) + P(A) \,.
\end{equation}
The only way this can happen is that $P(A)=0$, which is equivalent to $P\big( T(x_0) = \infty \big) =1$. This completes the proof.
\end{proof}
Theorem~\ref{q>0} confirms that, for small values of the noise $\sigma$ (the ones corresponding to $q \in [0,1]$, see~(\ref{phi})),
the solution of TSBE defines a bounded stochastic process, as intuitive. The same tools used in the proof of Theorem~\ref{q>0} 
do not allow however to draw any conclusion about the case $q<0$, where $|s(\pm 1)|$ are both finite. 
To study this case, it is useful to consider the mean exit time of the process from any $[a,b] \subseteq I$,
\begin{equation}\label{Mab}
M_{a,b}(x_0) = E \big[ T_{a,b}(x_0) \big] = \int_0^\infty P \big( T_{a,b}(x_0) > t \big) \,dt\,.
\end{equation}
If we denote by $p_{x_0}(t,x)$ the density of the process at time $t>0$, we can write the integrand on the RHS of~(\ref{Mab}) as
\begin{equation}
P \Big( T_{a,b}(x_0) > t \Big) = P \Big( x(t) \in (a,b) \Big) = \int_a^b p_{x_0}(x,t) \, dx\,.
\end{equation}
Therefore, we can exploit the Fokker-Planck equation for $p_{x_0}(x,t)$ to write down an ordinary differential equation for $M_{a,b}(x)$, view as function of $x$ only.
The Ordinary Differential Equation satisfied by $M_{a,b}(x)$ is as follows, cf.~\cite{Gardiner} for full details:
\begin{equation}\label{ODE_M}
 \varphi(x) \,M_{a,b}^\prime(x) + \frac{1}{2} \sigma^2(x)\,M_{a,b}^{\prime \prime}(x) = -1\,, \qquad x \in (a,b)\,.
\end{equation}
This is subject to the boundary conditions
\begin{equation}\label{bound_cond_M}
 M_{a,b}(a)= M_{a,b}(b) = 0\,,
\end{equation}
which immediately follows by the definition of $M_{a,b}$.

The analytic solution to the Boundary Value Problem~(\ref{ODE_M}) and (\ref{bound_cond_M}) is available, and reads as follows \cite{KS}:
\begin{equation}\label{M_int}
 M_{a,b}(x) = \int_a^b G_{a,b}(x,y) \; m(dy)\,,
\end{equation}
where $G_{a,b}$ is the following Green's function
\begin{equation}\label{Gab}
G_{a,b}(x,y) = \frac{\big[ s(x \land y) - s(a)   \big] \big[ s(b) - s(x \lor y) \big]}{s(b) - s(a)}\,, \qquad x,y \in [a,b]\,,
\end{equation}
and $m$ is the so-called \textit{speed measure} associated with the process (solution of the SDE with drift $\varphi$ e diffusion $\sigma$):
\begin{equation}  \label{m}
m(dy)=\frac{2 dy}{s^\prime(y)\sigma^2(y)}\,,\quad y \in I\,.
\end{equation}
Note again that the integrand of~(\ref{M_int}) does not depend on the particular choice of $c$ which has been made to define the scale function $s$ in~(\ref{s}).
In the Tsallis-Borland case, we can therefore choose $c=0$ and recover expression~(\ref{s_TSBE}) for $s(x)$.
We can now make a precise statement about the behavior of TSBE under the condition $q<0$.

\begin{theorem}\label{q<0}
 Consider TSBE~(\ref{TSBE}) with $q<0$. Then, for any initial condition $x_0 \in I=(-1,1)$, the solution
 $x(t)$ attains one of the boundaries of $I$ in finite time, with probability one.
 In other words,
 \begin{equation}
 P\big( T(x_0) < \infty \big) =1
 \end{equation}
 where $T(x_0):=\inf \left\{ t>0 \;\big|\;\, x(t)\notin I, \; x(0)=x_0 \right\}$ as in~(\ref{T(x_0)}).
\end{theorem}
\begin{proof}
Let us first recall from~(\ref{s_TSBE}) the expression of the scale function for TSBE as follows:
\begin{equation}\label{s(x)}
s(x)=\int_0^x \frac{dz}{ {\left( 1- z^2 \right)}^{1/\beta} }   \,,\qquad \beta=1-q>1\,.
\end{equation}
The assumption $q<0$ guarantees that both $|s(-1)|$ and $|s(+1)|$ are finite.
Therefore, we can extend Green's function $G_{a,b}$ in~(\ref{Gab}) to a maximal function $G_{-1,1}$ defined on the whole square ${[-1,1]}^2$:
\begin{equation}\label{G11}
G_{-1,1}(x,y) = \frac{\big[ s(x \land y) - s(-1) \big] \big[ s(1) - s(x \lor y) \big]}{2\,s(1)}\,, \qquad x,y \in [-1,1]\,.
\end{equation}
$G_{-1,1}$ is continuous on ${[-1,1]}^2$, with in particular $G(x_0,-1) = G(x_0,+1)=0$.
The average exit time from $I$ of the process $x(t)$ starting at $x_0$ can then be obtained from~(\ref{M_int}), 
by letting $a$ tend to $-1$ and $b$ tend to $1$. We have:
\begin{align}
  E\big[ T(x_0)\big] &= \int_{-1}^1 G_{-1,1}(x_0,y) \; m(dy)\nonumber \\
                     &= \frac{1}{\beta}\int_{-1}^1 \frac{\,G_{-1,1}(x_0,y)\,}{s^\prime(y)} \,dy \nonumber\\
                     &= \frac{1}{\beta}\int_{-1}^1 {(1- y^2)}^{1/\beta} \, G_{-1,1}(x_0,y)\,dy \,. \label{E(Tx0)}
\end{align}
Given the continuity of last integrand for all $y \in [-1,1]$,
we deduce that $E\big[ T(x_0)\big] < \infty$. In particular, this assures that $T(x_0)$ is almost surely finite, as it was to be proved.
\end{proof}

Theorem~\ref{q<0} therefore proves that the Tsallis-Stariolo-Borland process reaches one of the endpoints~$\pm 1$ in finite
time with probability one, if $q<0$.
Once this happens, the loss of regularity of the coefficients does not
guarantee that the solution of the SDE can be extended in a unique way. A more
in-depth analysis would indeed show that the uniqueness is lost in this
case. There is, in fact, a positive probability that the process develop
outside the bounded interval $I$ after attaining one of the boundaries, with
a consequent dispersion of the initial mass on the whole real line.
We will not provide a formal proof of these facts here. 
Indeed, we think that the most important and unexpected
result has already been shown in Theorem~\ref{q<0}, and consists in the reachability of the
boundaries under the condition $q<0$.

In Subsection~\ref{average_exit_time} we estimate, as a function of $q<0$, the average time that the process needs
to attain one of the boundaries $\pm 1$. In particular, this will provide 
the order of the speed at which $E\big[ T(x_0)\big]$ tends to infinity as $q \rightarrow 0^-$.

\subsection{Average exit time}\label{average_exit_time}

Formula~(\ref{E(Tx0)}) was used in the proof of Theorem~\ref{q<0} to show that the expected exit time of the process from $(-1,1)$
is finite if $q<0$, by a trivial argument of continuity of the integrand. Recall that, in that formula, we put
\begin{equation}
 \beta= 1-q >1 \quad \text{if } q<0\,.
\end{equation}
In the following, the functional dependence of $E\big[ T(x_0)\big]$ on the parameter $q$ will be made explicit.
Without loss of generality, we consider the case where the process starts from $x_0=0$. We therefore have:
\begin{equation}
 E\big[ T(0)\big] = \frac{1}{\beta}\int_{-1}^1 {(1- y^2)}^{1/\beta} \,G_{-1,1}(0,y) \;dy\,.
\end{equation}
From~(\ref{s(x)}) it immediately follows that $s(x)$ is an odd function. Thus, by~(\ref{G11})
\begin{equation}
 G_{-1,1}(0,y) =  \frac{\big[ s(0 \land y) - s(-1) \big] \big[ s(1) - s(0 \lor y) \big]}{2\,s(1)} = \frac{1}{2} \big[ s(1) - s(| y |) \big]\,,
\end{equation}
and
\begin{align}\label{E(T0)}
  E\big[ T(0)\big] &= \frac{1}{2 \beta}\int_{-1}^1 {(1- y^2)}^{1/\beta} \, \big[ s(1) - s(| y |) \big] \;dy \nonumber \\
                   &= \frac{1}{  \beta}\int_{0}^1 {(1- y^2)}^{1/\beta} \, \big[ s(1) - s(y) \big] \;dy \nonumber\\
                   &= \frac{1}{  \beta}\int_{0}^1 {(1- y^2)}^{1/\beta} \, \bigg[ \int_y^1 \frac{dz}{ {\left( 1- z^2 \right)}^{1/\beta} } \bigg] dy\,,
\end{align}
given the expression of $s$ in~(\ref{s(x)}). Now observe that the condition $\beta>1$ implies
\begin{equation}\label{1+y}
 1 \leq {{(1+y)}^{1/\beta}} \leq 2 \qquad \text{for all $y \in [0,1]$}
\end{equation}
\begin{equation}\label{1+z}
 \frac{1}{2} \leq \frac{1}{{(1+z)}^{1/\beta}} \leq 1 \qquad \text{for all $z \in [0,1]$} \,.
\end{equation}
In particular, from~(\ref{1+z}), it follows that 
\begin{align}
 \int_y^1 \frac{dz}{{\left(1-z^2\right)}^{1/\beta}} &= \int_y^1 \frac{1}{{\left(1+z\right)}^{1/\beta}} \, \frac{dz}{{\left( 1-z \right)}^{1/\beta}}\nonumber \\
                                                    &= C_1(\beta,y) \int_y^1 \frac{dz}{ {\left( 1- z \right)}^{1/\beta} }  \nonumber\\
                                                    &= C_1(\beta,y) \frac{\beta}{\beta -1} \, {(1-y)}^{\frac{\beta -1}{\beta}}\,, \label{s(1)-s(y)} 
\end{align}
where
\begin{equation}\label{C(b)}
 C_1(\beta,y) \in \left[ \frac{1}{2}, 1 \right] \quad \text{ for all $y \in [0,1]$, $\beta>1$.}
\end{equation}
Given~(\ref{s(1)-s(y)}), the average exit time $E\big[ T(0)\big]$ in~(\ref{E(T0)}) takes the following form:
\begin{align}
E\big[ T(0)\big] &= \frac{1}{ (\beta-1)} \int_{0}^1 {(1- y^2)}^{1/\beta} \, C_1(\beta,y) \, {(1-y)}^{1 - \frac{1}{\beta} }\, dy \nonumber \\[1mm]
                 &= \frac{1}{ (\beta-1)} \int_{0}^1 {(1+y)}^{1/\beta} \, C_1(\beta,y) \, (1-y) \, dy \nonumber \\[1mm]
                 &= \frac{{C_2}(\beta)}{ (\beta-1)} \int_{0}^1 (1-y) \, dy \;=\; \frac{{C_2}(\beta)}{2 (\beta-1)} \,, 
                 \qquad {C_2}(\beta) \in \textstyle \left[ \frac{1}{2}, 2 \right]
\end{align}
The bounds for $C_2(\beta)$ immediately follow by~(\ref{1+y})~and~(\ref{C(b)}). In terms of the original parameter $q=1-\beta <0$, we have
\begin{equation}\label{estim_average_time}
 E\big[ T(0)\big] = - \frac{C(q)}{q} \,, \qquad \frac{1}{4} \leq C(q) \leq 1\,.
\end{equation}
In particular, expression~(\ref{estim_average_time}) allows to deduce the asymptotic behavior of the average exit time 
in the two cases $q \rightarrow 0^-$ and $q \rightarrow - \infty$.
\begin{itemize}
 \item If the negative parameter $q$ approaches zero, then the average time needed to attain one of the boundaries tends to infinity,
 linearly in $1/|q|$:
 \begin{equation}
 E\big[ T(0)\big] \approx \frac{1}{\varepsilon} \qquad \text{for } q=-\varepsilon, \;\varepsilon\ll 1\,.
 \end{equation}
 \item The average time needed to attain one of the boundaries can be made arbitrarily small, as long as the 
 parameter $q$ is chosen (negative) large enough:
 \begin{equation}
  \lim_{q \rightarrow - \infty}  E\big[ T(0)\big] = 0\,.
 \end{equation}
\end{itemize}

\section{An (apparent) physical paradox, and a really bounded noise}\label{Sec_second_order}

Apparently, from a physical point of view, this means that the material
point could eventually reach and overcome the boundaries of the
infinite-height well, as a pure consequence of sufficiently large stochastic
fluctuations. However, this apparent paradox has an easy explanation: the
paradox simply comes from the overdamped approximation, which in
this particular case led to an unphysical result. \\
As stressed by H{\"a}nggi and Jung \cite{HJ,hjucna}, the large friction approximation is equivalent 
to the condition of validity of the Smoluckowski approximation, which reads as follows \cite{Beck,hjucna,HJ}:
\begin{equation}\label{smoluch_rule0}
\gamma \gg \sqrt{D} \,\left| \frac{d}{dx}Log(|F(x)|) \right|,
\end{equation}
where $D$ is the diffusion coefficient, and $F(x)$ is the conservative force the point is subject to.
In our case, this yields:
\begin{equation}\label{smoluch_rule1}
1 \gg \sqrt{2(1-q)} \,\left| \frac{1}{x} + \frac{2x}{1-x^2} \right|.
\end{equation}
It is interesting to note that the constraint (\ref{smoluch_rule1}) is violated 
not only for $x$ close to $-1$ and $1$, as it is intuitive, but also close to $0$.\\
We are however going to show that the infinite potential boundaries cannot be overcome 
in the original full Newton's equation
\begin{equation}\label{x''}
m \ddot x=- \dot x+ \left( \frac{\sgn\left(
x-1\right) }{\left\vert x-1\right\vert }+\frac{\sgn\left( x+1\right) }{%
\left\vert x+1\right\vert }\right) + \sqrt{2 (1-q) }\,\xi ( t),
\end{equation}%
representing the motion of the point $P$ (see~(\ref{newton1})) under forces (\ref{FlTSBE}) and (\ref{FrTSBE}). 
In this regard, let the initial position
be $x\left( 0\right) \in \left( -1,1\right) $ and let it be any initial
velocity $v\left( 0\right) \in \left( -\infty ,\infty \right) $. Let us
prove that the barriers $\pm 1$ are never reached, independently of the
value of $q$ ($q<1$, of course).

\begin{theorem}\label{Thm:second_order}
For all $q<1$, the solution of~(\ref{x''}) with initial condition $x_0 \in I=(-1,1)$ exists globally in time, is unique
and is contained in $I$ for all times.
\end{theorem}

\begin{proof}
Before we start, we explain the idea. In the deterministic
case, recalled for convenience in Step~1, the global energy is decreasing, because of the viscosity. Since we start from an initial condition with finite
energy, the infinite potential barriers cannot be reached, because the
energy must remain finite. \\
The extension of this simple argument to the stochastic case requires a
proof. Indeed, the additive noise introduces energy in the average, as the
energy balance inequality~(\ref{energy balance}) shows. Thus one has to
prove that this injected energy is not sufficient to overcome the infinite
potential barriers. This is done in Step~2. One detail is however delicate,
namely taking expected value of the It\^{o} integral when we only know it is
a local martingale (namely we do not know a priori that the integrand is
square integrable in all variables). Step 2 is completed a little bit
formally by using the fact that this It\^{o} integral has zero expectation

Then, in Step~3, we show how to make it rigorous. 

\textbf{Step 1}. Let us rewrite equation~(\ref{x''}) in position-velocity coordinates:%
\begin{align}
\dot x& =v \\
m \dot v& =- v+  \left( \frac{\sgn\left( x-1\right) }{%
\left\vert x-1\right\vert }+\frac{\sgn\left( x+1\right) }{\left\vert
x+1\right\vert }\right) + \sqrt{2\beta }\,\xi (t) ,
\end{align}%
where we denote by $\beta$ the positive constant $1-q$.\\
Let us explain first the idea in the deterministic case $\beta =0$ (the
result in this case is well known). 
The potential energy, kinetic energy, and total energy read as follows: 
\begin{align}
U( x) & =- \left( \log \left\vert x-1\right\vert +\log \left\vert x+1\right\vert \right)  \\
K( v) & =\frac{m}{2}v^{2} \\
E(x,v) & = U(x) + K(v)
\end{align}%
respectively. Notice that%
\begin{equation}
U^{\prime }=- \left( \frac{\sgn\left( x-1\right) }{\left\vert
x-1\right\vert }+\frac{\sgn\left( x+1\right) }{\left\vert x+1\right\vert }%
\right) .
\end{equation}

Let $x\left( t\right) $ be a solution, with $x\left( 0\right) \in I$, $I=\left(
-1,1\right) $, defined on some interval $[0,T_{0})$ (a local in time unique
solution exists since the coefficients of the equation are locally Lipschitz
continuous on $I$). By classical arguments of analysis one
can consider the maximal interval of time $[0,T_{\max })$ where the solution
exists unique and belongs to $I$. We then have two possibilities for $T_{\max}$: 
either $T_{\max }=+\infty$, or $T_{\max }<\infty $ and $%
\lim_{t\rightarrow T_{\max }}x\left( t\right) $ is either 1 or -1. 
\\In order to prove that $T_{\max }=+\infty$, let us
show that, for all $t\in \lbrack 0,T_{\max })$, 
\begin{equation}\label{ineq4}
\min \left( \left\vert x\left( t\right) -1\right\vert ,\left\vert x\left(
t\right) +1\right\vert \right) \geq \exp \left( -E\left( 0\right)-\log 2\right) >0.
\end{equation}%
This implies $T_{\max }=+\infty $, because under inequality~(\ref{ineq4}) $%
\lim_{t\rightarrow T_{\max }}x\left( t\right) $ cannot be equal to 1 or -1.
\\
On $[0,T_{\max })$ we have $x\left( t\right) \in \left( -1,1\right) $. Since 
\begin{equation}
\frac{d}{dt}E=U^{\prime }\dot x+K^{\prime } \dot v=U^{\prime
} v + m v \left( -\frac{1 }{m}v-\frac{1}{m}U^{\prime }\right)
=- v^{2}\leq 0\,,
\end{equation}%
we have $E\left( t\right) \leq E\left( 0\right) <\infty $ and in particular, 
\begin{equation}
-\log \left\vert x\left( t\right) -1\right\vert -\log \left\vert x\left(
t\right) +1\right\vert \leq E\left( 0\right).
\end{equation}%
The property $x\left( t\right) \in \left( -1,1\right) $ implies $-\log
\left\vert x\left( t\right) +1\right\vert \geq -\log 2$, hence%
\begin{equation}
-\log 2-\log \left\vert x\left( t\right) -1\right\vert \leq E\left(
0\right)
\end{equation}%
which implies 
\begin{equation}
\left\vert x\left( t\right) -1\right\vert \geq \exp \left( -E\left(
0\right)-\log 2\right) .
\end{equation}%
The inequality $\left\vert x\left( t\right) +1\right\vert \geq \exp \left( -%
E\left(0\right)-\log 2\right) $ is similar. We have proved
the claim of the theorem in the deterministic case.

\textbf{Step 2}. Let us now give the proof in the stochastic case, $\beta
\neq 0$. As far as the solution has the property $x\left( t\right) \in \left(
-1,1\right) $, it lives in a region of $\left( x,v\right) $ space where the
coefficients of the equation are locally Lipschitz continuous. Hence a
unique maximal solution exists, maximal with the property $x\left( t\right)
\in \left( -1,1\right) $, on a random time interval $[0,T_{\max })$. We have
to prove that $P\left( T_{\max }=+\infty \right) =1$. When $T_{\max }<\infty $,
one has $\lim_{t\rightarrow T_{\max }}x\left( t\right) =\pm 1$.
\\
Let $\left( x,v\right) $ be the maximal solution, on $[0,T_{\max })$. By It%
\^{o} formula, on $[0,T_{\max })$, 
\begin{align}
\dot E& =U^{\prime }\dot x+K^{\prime }\dot v+\frac{1}{2}%
K^{\prime \prime }\frac{2\beta }{m^{2}} \nonumber \\
& =U^{\prime }v + m v \left( -\frac{1 }{m}v-\frac{1}{m}U^{\prime }+%
\frac{ \sqrt{2\beta }}{m}\xi \right) +m\frac{\beta }{m^{2}}
\nonumber \\
& =- v^{2}+\frac{\beta }{m}+\sqrt{2\beta }v\xi .
\end{align}%
Representing $\xi $ as the derivative of Brownian motion, $\xi =\dot B $, we get%
\begin{equation}
E\left( t\right) \leq E\left( 0\right) +\frac{\beta }{m}%
t+\int_{0}^{t} \sqrt{2\beta }\,v(s) dB(s) .
\label{energy balance}
\end{equation}%
In particular, always on $[0,T_{\max })$, the two following inequalities hold:
\begin{equation}
-\log \left\vert x\left( t\right) -1\right\vert -\log \left\vert x\left(
t\right) +1\right\vert \leq E\left( 0\right)+\frac{
\beta }{m}t+\int_{0}^{t}\sqrt{2\beta }\,v(s) dB(s)  
\label{ineq 1}
\end{equation}%
\begin{equation}
\frac{m}{2}v^{2}\left( t\right) \leq E\left( 0\right) +\frac{%
\beta }{m}t+\int_{0}^{t}\sqrt{2\beta }\,v(s) dB(s) .  \label{ineq 2}
\end{equation}%
From~(\ref{ineq 2}) we deduce (notice that $t\wedge T_{\max }\leq t$) 
\begin{equation}
E\left[ v^{2}\left( t\wedge T_{\max }\right) \right] \leq \frac{2E\left(
0\right) }{m}+\frac{2 \beta }{m^{2}}t  \label{est 1}
\end{equation}%
since the expectation of the It\^{o} integral is zero (the rigorous proof of
this inequality requires an argument of stopping times and thus it is
postponed to Step 3 below; also the proper definition of $v^{2}\left(
t\wedge T_{\max }\right) $ is given there).

Then, from~(\ref{ineq 1}) and Doob's inequality for the It\^{o}
integral, we deduce that, for each given deterministic time $T>0$,%
\begin{align}
& E\left[ \sup_{t\in \lbrack 0,T\wedge T_{\max })}\big( \log \left\vert
x\left( t\right) -1\right\vert +\log \left\vert x\left( t\right)
+1\right\vert \big) ^{2}\right]  \nonumber\\
& \leq 2\left( E\left( 0\right)+\frac{\beta }{m}%
T\right) ^{2}+16\beta \int_{0}^{T}E\left[ v^{2}(s) 1_{s\leq
T_{\max }}\right] ds \nonumber \\
&\leq 2\left( E\left( 0\right)+\frac{\beta }{m}%
T\right) ^{2}+16\beta T\left( \frac{2E\left( 0\right) }{m}+\frac{2\beta }{m^{2}}T\right) <\infty   \label{est 2}
\end{align}%
(also the proof of this claim is given in detail in Step 3 below). This
implies 
\begin{equation}
\sup_{t\in \lbrack 0,T\wedge T_{\max }))}\big( \log \left\vert x\left(
t\right) -1\right\vert +\log \left\vert x\left( t\right) +1\right\vert
\big) ^{2}<\infty 
\end{equation}%
with probability one. Then necessarily $T_{\max }=+\infty $, because in the
opposite case, from $\lim_{t\rightarrow T_{\max }}x\left( t\right) =\pm 1$,
the supremum would be infinite. We have completed the proof that $T_{\max
}=+\infty $, with probability one, which includes in particular the claim
that $x\left( t\right) \in \left( -1,1\right) $ for all $t\geq 0$, with
probability one.

\textbf{Step 3}. Let us prove (\ref{est 1}). Let $\tau _{n}$ be an
increasing sequence of finite stopping times which converges almost surely to $%
T_{\max }$ from below. Let $\sigma _{n}$ be the stopping time defined as%
\begin{equation}
\sigma _{n}=\inf \left\{ t\geq 0:\left\vert v(t) \right\vert
>n\right\} \wedge \tau _{n}
\end{equation}%
($\sigma _{n}=\tau _{n}$ if the set is empty). From inequality (\ref{ineq 2}%
) we have%
\begin{equation}
\frac{m}{2}v^{2}\left( t\wedge \sigma _{n}\right) \leq E\left( 0\right)
+\frac{\beta }{m}\left( t\wedge \sigma _{n}\right)
+\int_{0}^{t} \sqrt{2\beta }\,v(s) 1_{s\leq \sigma
_{n}}dB( s) .
\end{equation}%
Since $\left\vert v( s) 1_{s\leq \sigma _{n}}\right\vert \leq n$,
the It\^{o} integral above is a martingale and thus its average is zero.
Therefore, since $t\wedge \sigma _{n}\leq t$, 
\begin{equation}
\frac{m}{2}E\left[ v^{2}( t\wedge \sigma _{n}) \right] \leq
E\left( 0\right) +\frac{\beta }{m}t
\end{equation}%
namely $E\left[ v^{2}( t\wedge \sigma _{n}) \right] \leq \frac{%
2E\left( 0\right) }{m}+\frac{2 \beta }{m^{2}}t$. By Fatou lemma%
\begin{equation}
E\left[ \underset{n\rightarrow \infty }{\lim \inf }v^{2}\left( t\wedge
\sigma _{n}\right) \right] \leq \frac{2E\left( 0\right) }{m}+\frac{2\beta }{m^{2}}t.  \label{est 3}
\end{equation}%
One can check that $\lim_{n\rightarrow \infty }\sigma _{n}=T_{\max }$. If $%
T_{\max }=\infty $, $\underset{n\rightarrow \infty }{\lim \inf }v^{2}(
t\wedge \sigma _{n}) =v^{2}(t) $. If $T_{\max }<\infty $,
namely when $\lim_{n\rightarrow \infty }t\wedge \sigma _{n}=t\wedge T_{\max }
$, we do not know a priori that $v$ can be prolonged with continuity at time 
$T_{\max }$ (the solution $\left( x,v\right) $ is defined only on the
maximal interval $[0,T_{\max })$). But $\underset{n\rightarrow \infty }{\lim
\inf }v^{2}( t\wedge \sigma _{n}) <\infty $ with probability one,
by (\ref{est 3}). Thus we define $v^{2}( t\wedge T_{\max }) $ as
this lim inf. Thus (\ref{est 3}) is the correct meaning of (\ref{est 1}).
\\[1mm]
Let us now prove (\ref{est 2}). Let $\sigma _{n}^{\prime }$ be the stopping
time defined as%
\begin{equation}
\sigma _{n}^{\prime }=\inf \left\{ t\geq 0:\left\vert x\left( t\right)
\right\vert >1-\frac{1}{n}\right\} \wedge \sigma _{n}
\end{equation}%
($\sigma _{n}^{\prime }=\sigma _{n}$ if the set is empty). From (\ref{ineq 1}%
) we have (since $t\wedge \sigma _{n}^{\prime }\leq t$)%
\begin{equation}
-\log \left\vert x\left( t\wedge \sigma _{n}^{\prime }\right) -1\right\vert
-\log \left\vert x\left( t\wedge \sigma _{n}^{\prime }\right) +1\right\vert
\leq E\left( 0\right)+\frac{ \beta }{m}t+\int_{0}^{t}%
\sqrt{2\beta }\,v(s) 1_{s\leq \sigma _{n}^{\prime }}dB(s) .
\end{equation}%
Again $\left\vert v( s) 1_{s\leq \sigma _{n}^{\prime
}}\right\vert \leq n$, so $\int_0^t v(s) 1_{s\leq \sigma _{n}^{\prime }}dB(s)$ is a martingale. 
Hence, by Doob's inequality and It\^{o} isometry formula,%
\begin{equation}
E\left[ \sup_{t\in \left[ 0,T\right] }\left( \log \left\vert x\left( t\wedge
\sigma _{n}^{\prime }\right) -1\right\vert +\log \left\vert x\left( t\wedge
\sigma _{n}^{\prime }\right) +1\right\vert \right) ^{2}\right] \leq 2\left( 
E\left( 0\right)+\frac{\beta }{m}T\right)
^{2}+16\beta \int_{0}^{T}E\left[ v^{2}(s) 1_{s\leq \sigma
_{n}^{\prime }}\right] ds.  \label{ineq 3}
\end{equation}%
One can check that $\lim_{n\rightarrow \infty }\sigma _{n}^{\prime }=T_{\max
}$. By monotone convergence, the right-hand-side of~(\ref{ineq 3}) converges to 
\begin{equation}
2\left( E\left( 0\right)+\frac{\beta }{m}T\right)
^{2}+16\beta \int_{0}^{T}E\left[ v^{2}\left( s\right) 1_{s\leq T_{\max }}\right] ds\,.
\end{equation}%
Moreover
\begin{equation}
E\left[ v^{2}\left( s\right) 1_{s\leq T_{\max }}\right] \leq E\left[
v^{2}\left( s\right) 1_{s\leq T_{\max }}\right] +E\left[ v^{2}\left( T_{\max
}\right) 1_{s>T_{\max }}\right] =E\left[ v^{2}\left( s\wedge T_{\max
}\right) \right] ,
\end{equation}%
thus
\begin{align}
16\beta \int_{0}^{T}E\left[ v^{2}\left( s\right) 1_{s\leq T_{\max }} \right] ds
\leq 16\beta \int_{0}^{T}E\left[ v^{2}\left( s\wedge T_{\max }\right) \right] ds
\leq 16\beta T\left( \frac{2E\left( 0\right) }{m}+\frac{2\beta }{m^{2}}T\right) 
\end{align}%
by~(\ref{est 1}). 
We have therefore seen that the right-hand-side of~(\ref{ineq 3}) is bounded by
\begin{equation}
2\left( E\left( 0\right)+\frac{\beta }{m}T\right)^{2}
+ 16\beta T\left( \frac{2E\left( 0\right) }{m}+\frac{2 \beta }{m^{2}}T\right) \,.
\end{equation}
It remains to understand the limit of the left-hand-side
of~(\ref{ineq 3}). One has%
\begin{align}
& \sup_{t\in \left[ 0,T\right] }\big( \log \left\vert x\left( t\wedge
\sigma _{n}^{\prime }\right) -1\right\vert +\log \left\vert x\left( t\wedge
\sigma _{n}^{\prime }\right) +1\right\vert \big) ^{2} \\
& =\sup_{t\in \lbrack 0, T\wedge \sigma _{n}^{\prime })}\big( \log \left\vert
x\left( t\right) -1\right\vert +\log \left\vert x\left( t\right)
+1\right\vert \big) ^{2},
\end{align}%
hence again we may apply monotone convergence and get that the
left-hand-side of (\ref{ineq 3}) converges to%
\begin{equation}
E\left[ \sup_{t\in \lbrack 0,T\wedge T_{\max })}\big( \log \left\vert x\left(
t\right) -1\right\vert +\log \left\vert x\left( t\right) +1\right\vert
\big) ^{2}\right] .
\end{equation}%
This proves (\ref{est 2}) and completes the proof of the theorem.
\end{proof}

The above Theorem~\ref{Thm:second_order}, despite the quite technical proof, unequivocally shows the following: 
under the action of the above-described potential $U(x)$ as well as the viscous force, and by fully taking into account the mass $m$ of the point P,
the motion remains bounded independently of the particular value of the parameter $q<1$.
For the sake of completeness, it is worth mentioning that the stationary probability density 
of the vectorial process $(x,v)$ is in the classical Boltzmann-like form:
\begin{equation}
\rho_{ss}(x,v)= A \times \exp \left[ -\frac{1}{1-q}\Big( U(x) + K(v) \Big) \right] ,
\end{equation}
yielding:
\begin{equation}
 \rho_{ss}(x,v)= A \times e^{ -\frac{1}{1-q} \frac{m v^2}{2}}(1-x^2)^{1/(1-q)} .
\end{equation}

\section{A parametric extension of the TSB model leading to a first order SDE
with bounded solutions} \label{Sec_alpha_family}

In this section we propose a one-parameter family of noises that includes as
particular case the TSB noise. Indeed, let us consider the following family
of forces $\varphi^{\alpha}(x) $ which depend on a parameter $\alpha>0$: 
\begin{equation}
\varphi^{\alpha}(x) = \varphi_l^{\alpha}(x) + \varphi_r^{\alpha}(x),
\end{equation}
where:
\begin{equation}  \label{FlTSBEalpha}
\varphi_l^{\alpha}(x)=\frac{\sgn(x+1)}{|x+1|^{\alpha}};
\end{equation}
\begin{equation}\label{FrTSBEalpha}
\varphi_r^{\alpha}(x) = \frac{\sgn(x-1)}{|x-1|^{\alpha}};
\end{equation}
According to the notations of Section~\ref{Sec_basic}, the corresponding full Newton's equation
and the overdamped approximated equation ($m \ll 1$) read as follows, respectively:
\begin{equation}\label{newton_alpha}
 m \ddot x =- \dot x  + \varphi^{\alpha}(x) +  \sqrt{2(1-q)} \,\xi(t),
\end{equation}
and
\begin{equation}\label{SDE_alpha}
 \dot x  = \varphi^{\alpha}(x) + \sqrt{2(1-q)} \xi(t)\,.
\end{equation}
The potential associated with~(\ref{newton_alpha}) takes therefore the following form: 
\begin{equation}\label{U_alpha}
U^{\alpha}(x) =
\begin{cases}
\frac{1}{\alpha-1} \left(|x+1|^{1-\alpha}+|x-1|^{1-\alpha} \,-2\right) & %
\mbox{if $\alpha \neq 1$ and $\alpha>0$} \\[2mm]
- \log(1-x^2) & \mbox{if } \alpha =1
\end{cases}
\end{equation}
where the constant has been chosen so that $U^{\alpha}(0)=0$ for all $\alpha$.
\begin{figure}[t!]
\centering
\includegraphics[scale=0.5]{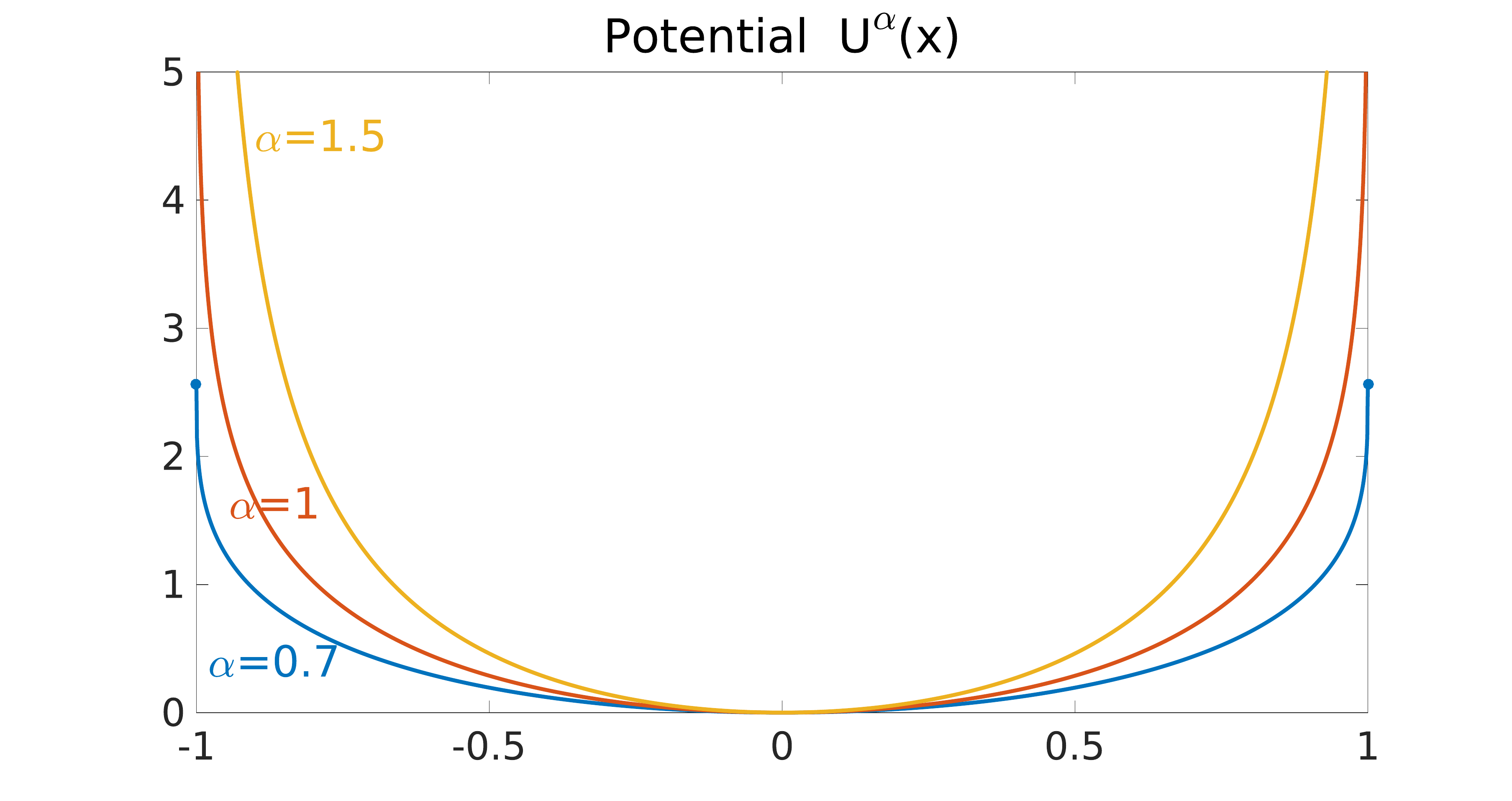}
\caption{Graph of the potential $U^\alpha$ in~(\ref{U_alpha}) for $\alpha=0.7$, $\alpha=1$, and $\alpha=1.5$. 
The potential is bounded for $\alpha<1$, and unbounded for $\alpha \geq 1$. However, the case $\alpha=1$ only yields
a logarithmic growth near the boundaries, while the case $\alpha>1$ yields a polynomial growth.}
\label{fig_potential_alpha}
\end{figure}
Notice that, for $\alpha \in (0,1)$, the above potential is
bounded for finite values of $x$, so it does not form a potential well (Figure~\ref{fig_potential_alpha}, blue line). 
\newline
The case $\alpha =1$ has been studied in detail in Sections~\ref{Sec_q<0}~and~\ref{Sec_second_order}: the
potential forms in fact an infinite well but, in the case of the approximated equation~(\ref{SDE_alpha}), 
the behavior of the solution further depends on the parameter $q$.
As far as the full Newton's equation is concerned, instead, the solution remains bounded independently of the value of $q$.
\newline
Finally, the case $\alpha>1$ remains to be investigated.
Preliminarily, we note that also in this case the potential $U^{\alpha}(x)$ forms
an infinite potential well, as illustrated in Figure~\ref{fig_potential_alpha}.
Second, we note that the proof of Theorem~\ref{Thm:second_order} can be used to show that
the family of Newton's equations in~(\ref{newton_alpha}), i.e.
\begin{equation}
m \ddot x=- \dot x  +  \frac{\sgn(x+1)}{\left\vert
x+1\right\vert ^{\alpha}} +  \frac{ \sgn(x-1)}{\left\vert x-1\right\vert ^{\alpha}}+%
\sqrt{2(1-q)}\,\xi( t) ,
\end{equation}
gives rise to solutions which never leave the interval $I=(-1,1)$ for all positive times.
Thus, in the case $\alpha>1$, it remains to check whether the same holds true also for the first order SDE~(\ref{SDE_alpha}).

However, by looking back at the proof of Theorem~\ref{q>0}, we see that a sufficient condition for the process
not to reach the boundaries of its state space is that the scale function $s^\alpha(x)$ associated with~(\ref{SDE_alpha}) explode at the boundaries.
So, let us simply check that $\vert s^\alpha(\pm 1) \vert = \infty$ under the hypothesis $\alpha>1$, where
\begin{equation}
 s^\alpha(x)= \int_0^x \exp \left( - \int_0^y 2\:\frac{\varphi^\alpha(z)}{\sigma^2(z)}\;dz\;\right)\,dy\,.
\end{equation}
For the sake of simplicity, define $\beta=1-q$. We have:
\begin{align}
 s^\alpha(1) &= \int_0^1 \exp \left( - \frac{1}{\beta} \int_0^y \varphi^{\alpha}(z) \;dz\;\right)dy \nonumber \\[2mm]
             &= \int_0^1 \exp \left( \frac{U^{\alpha}(y)}{\beta} \;\right)dy \nonumber \\[2mm]
             &\!\!\stackrel{(\ref{U_alpha})}{=} e^{-2C} \int_0^1 \exp \left[ \frac{C}{ {(1+y)}^{\alpha-1} } +\frac{C}{ {(1-y)}^{\alpha-1} } \right]dy\,,
             \qquad C = \frac{1}{\beta(\alpha-1)} \nonumber \\[2mm]
             &\geq e^{-2C} \int_0^1 \exp \left[ \frac{C}{ {(1-y)}^{\alpha-1} } \right]\,dy = + \infty\,,
\end{align}
where the last equality precisely holds because $\alpha>1$. Likewise, $s^\alpha(-1) = -\infty$. \\
So, we can safely conclude that the boundaries $\pm 1$ are not reached if $\alpha>1$, 
and that the process remains therefore bounded in this case, without any further assumption on the magnitude of the constant diffusion $\sigma$.
In particular, when $\alpha>1$, the process $x^\alpha(t)$ solution of~(\ref{SDE_alpha}) is ergodic for any value of the parameter $q<1$.
Its stationary density can be easily derived as a time-invariant solution of the Fokker-Planck equation, which (up to normalisation constant) yields
\begin{align}\label{stat_sensity_alpha}
p^\alpha(x;q) &= \exp \left( -\frac{U^\alpha(x)}{1-q} \right) \nonumber\\
            &= \exp \left[ -\frac{1}{(1-q) (\alpha-1)} \left(\frac{1}{{(1+x)}^{\alpha -1}} + \frac{1}{{(1-x)}^{\alpha -1}} \right) \right], \quad \alpha>1\,.
\end{align}
The mass of this density moves away from the boundaries of $I=(-1,1)$ as the value of $\alpha$ increases, as Figure~\ref{fig_density_alpha} shows. 

\begin{figure}[t!]
\centering
\includegraphics[scale=0.5]{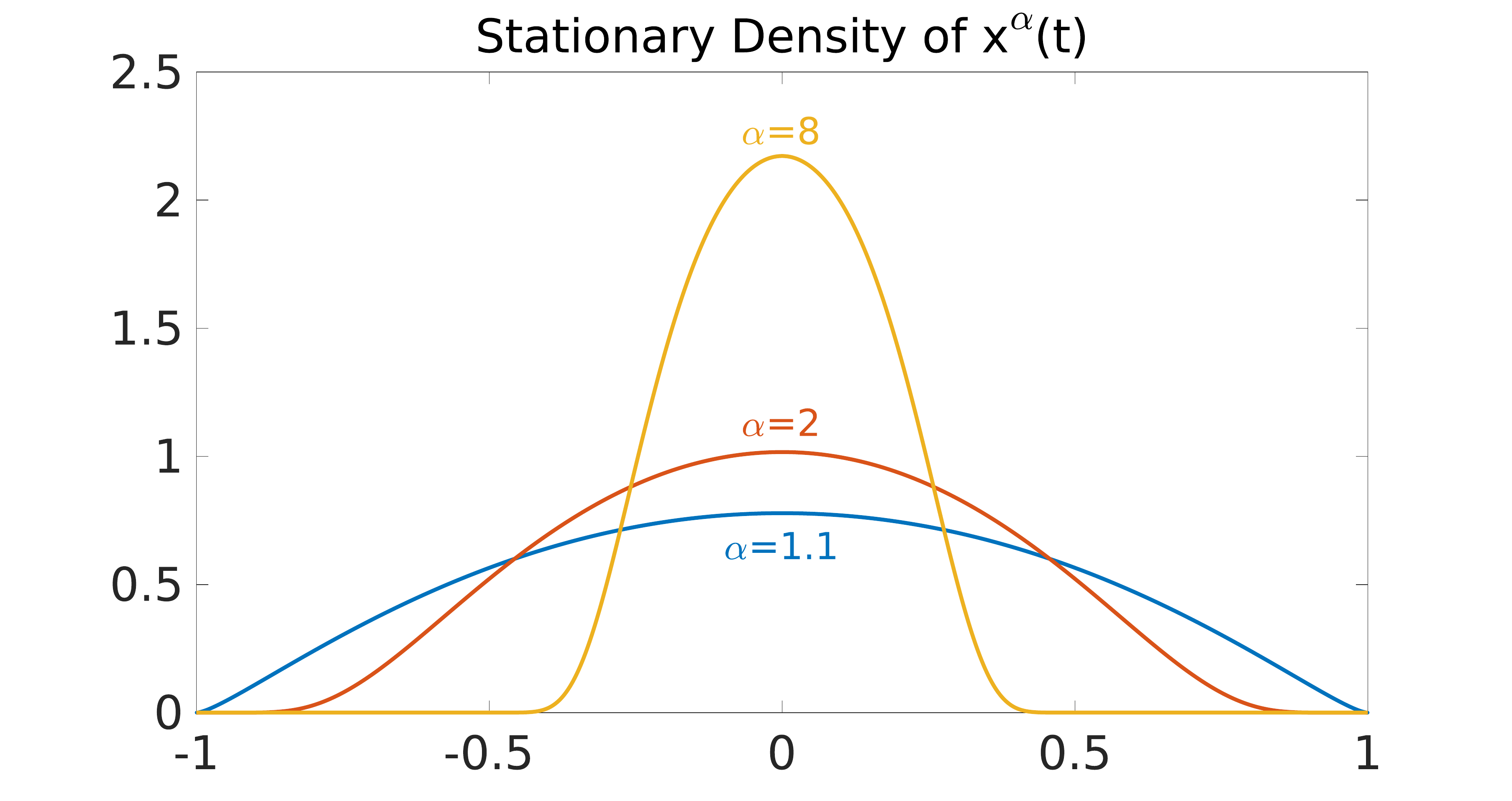}       
\caption{Normalised stationary density of the process solution of~(\ref{SDE_alpha}), for three different values of $\alpha>1$. 
The value $q=0$ has been chosen.
Notice how a lower mass is attached to the boundaries $\pm 1$ as long as $\alpha$ 
increases from just above 1 to greater values.}
\label{fig_density_alpha}
\end{figure}
\noindent
Summarizing the results of this section, from an heuristic point of view we may
say that the case $\alpha=1$ is at the interface between potentials with and without
infinite wells, which yield bounded and unbounded solutions, respectively.
The potential $U^{\alpha=1}(x)$ is itself infinite at $x=\pm 1$,
but its growth is very slow since it is logarithmic. 
As a consequence, the boundedness of the solutions obtained under the overdamped approximation
depends on the magnitude of the stochastic perturbation the point $P$ of small mass is subject to.

\section{Concluding Remarks}\label{Sec_conclusions}

In this work we have showed that the Tsallis-Stariolo-Borland SDE, despite having an
apparent stringent physical interpretation of overdamped stochastic motion
of a point $P$ in an infinite potential well, is able to generate unbounded noises for sufficiently large
diffusion coefficient (namely for negative Tsallis parameter $q$). The explanation of this
apparently unphysical and anti-intuitive behavior is that the overdamped
first order SDE is a result of the overdamped approximation,
which fails in our case.
Indeed, we have showed that the full Newtonian equation describing the motion of the material
point $P$ is able to generate a genuinely bounded stochastic process for the
position $x(t)$ of the particle.\\ 
We have also showed that the TSB case is at the
interface between finite and infinite potential wells in a family of SDEs
with potentials $U^{\alpha}(x)$ which depend on a real positive parameter $\alpha$. 
The properties of this family suggest that TSB potential might be a mathematical artifact
separating two more physical scenarios, where for $\alpha \in (0,1)$ the
motion is unbounded due to the boundedness of the potential for finite $x$,
and for $\alpha>1$ the motion is bounded both in presence and in absence of the overdamped approximation.

\bibliographystyle{unsrt}
\bibliography{Refer}

\end{document}